\documentclass[11pt]{article}
\pdfoutput=1
\usepackage[utf8]{inputenc}
\usepackage[T1]{fontenc}
\usepackage[margin=1in]{geometry}
\usepackage{amsmath,amsthm,amssymb,thmtools,thm-restate,booktabs,color,doi,graphicx,latexsym,url,xcolor,xspace}
\usepackage[numbers,sort&compress]{natbib}
\usepackage{microtype,hyperref}
\usepackage[inline]{enumitem}
\setlist[itemize]{label=--}
\setlist[enumerate]{label=(\arabic*),labelindent=\parindent,leftmargin=*}

\usepackage{sectsty}
\allsectionsfont{\boldmath}

\definecolor{citecolor}{HTML}{0000C0}
\definecolor{urlcolor}{HTML}{000080}

\hypersetup{
    colorlinks=true,
    linkcolor=black,
    citecolor=citecolor,
    filecolor=black,
    urlcolor=urlcolor,
    pdftitle={Efficient load-balancing through distributed token dropping},
    pdfauthor={Sebastian Brandt, Barbara Keller, Joel Rybicki, Jukka Suomela and Jara Uitto}
}

\usepackage{cleveref}
\usepackage{framed}

\newcommand{\local}{\ensuremath{\mathsf{LOCAL}}\xspace}
\newcommand{\congest}{\ensuremath{\mathsf{CONGEST}}\xspace}

\newtheorem{theorem}{Theorem}
\newtheorem{lemma}[theorem]{Lemma}

\newtheorem{definition}[theorem]{Definition}

\DeclareMathOperator{\indegree}{indegree}

\DeclareMathOperator{\poly}{poly}
\DeclareMathOperator{\polylog}{polylog}

\newenvironment{myabstract}
{\list{}{\listparindent 1.5em%
		\itemindent    \listparindent
		\leftmargin    1cm
		\rightmargin   1cm
		\parsep        0pt}%
	\item\relax}
{\endlist}

\newenvironment{mycover}
{\list{}{\listparindent 0pt
		\itemindent    \listparindent
		\leftmargin    1cm
		\rightmargin   1cm
		\parsep        0pt}%
	\raggedright
	\item\relax}
{\endlist}

\newcommand{\myemail}[1]{\,$\cdot$\, {\small #1}}
\newcommand{\myaff}[1]{\,$\cdot$\, {\small #1}\par\medskip}

\title{Efficient Load-Balancing through Distributed Token Dropping}

\begin{document}

\begin{mycover}
	{\huge\bfseries\boldmath Efficient Load-Balancing through Distributed Token Dropping\par}

	\bigskip
	\bigskip

\textbf{Sebastian Brandt}
\myemail{brandts@ethz.ch}
\myaff{ETH Zurich}

\textbf{Barbara Keller}
\myemail{barbara.keller@aalto.fi}
\myaff{Aalto University}

\textbf{Joel Rybicki}
\myemail{joel.rybicki@ist.ac.at}
\myaff{IST Austria}

\textbf{Jukka Suomela}
\myemail{jukka.suomela@aalto.fi}
\myaff{Aalto University}

\textbf{Jara Uitto}
\myemail{jara.uitto@aalto.fi}
\myaff{Aalto University}

\bigskip
\end{mycover}

\medskip
\begin{myabstract}
  \noindent\textbf{Abstract.}
    We introduce a new graph problem, the \emph{token dropping game}, and we show how to solve it efficiently in a distributed setting. We use the token dropping game as a tool to design an efficient distributed algorithm for \emph{stable orientations} and more generally for \emph{locally optimal semi-matchings}. The prior work by Czygrinow et al.\ (DISC 2012) finds a stable orientation in $O(\Delta^5)$ rounds in graphs of maximum degree $\Delta$, while we improve it to $O(\Delta^4)$ and also prove a lower bound of $\Omega(\Delta)$. For the more general problem of locally optimal semi-matchings, the prior upper bound is $O(S^5)$ and our new algorithm runs in $O(C \cdot S^4)$ rounds, which is an improvement for $C = o(S)$; here $C$ and $S$ are the maximum degrees of customers and servers, respectively.
\end{myabstract}

\thispagestyle{empty}
\setcounter{page}{0}
\newpage

\section{Introduction}

We consider efficient distributed algorithms for assignment problems. The task is to assign each customer to one adjacent server, and the customers prefer servers with a low load, i.e., few other customers. We are interested in finding a \emph{stable assignment}, that is, an assignment in which no customer has incentive to unilaterally switch servers. The stable assignment problem, known as the \emph{locally optimal semi-matching} problem, was studied in the distributed setting by \citet{Czygrinow2012}.

We start with a restricted version of the problem, known as \emph{stable orientation}, which is a special case in which all customers can choose between exactly two servers, and we will show how the same ideas generalize to the stable assignment problem.

We approach these problems by introducing a new graph problem called the \emph{token dropping game}.  We show how to solve the token dropping game efficiently in the distributed setting, and how an efficient solution to the token dropping game can be employed to solve the stable orientation and the stable assignment problem efficiently.
Furthermore, we prove lower bounds for the token dropping game as well as for the stable assignment problem.

The main technical contribution is the improvement of the distributed round complexity of stable \emph{orientations} from $O(\Delta^5)$ to $O(\Delta^4)$, and the lower bound of $\Omega(\Delta)$ for graphs of maximum degree $\Delta$, both of these in the standard \local~\cite{Linial1992,Peleg2000} model of distributed computing.
We additionally show that the more general problem of finding stable \emph{assignments} can be solved in $O(C \cdot S^4)$ time, where $C$ and $S$ are the maximum degrees of customers and servers, respectively. \Cref{table:summary} gives as summary of our new bounds and prior bounds.

\subsection{Stable orientations}

Consider the following problem on a graph:
\begin{framed}
\noindent
All edges are oriented, and each edge wants to selfishly minimize the indegree of the node to which it is pointing, by flipping or keeping its orientation.
\end{framed}
\noindent More precisely, we say that an oriented edge $e = (u,v)$ is \emph{happy} if 
\[
\indegree(v) \le \indegree(u) + 1,
\]
that is, turning the orientation of $e$ from $(u,v)$ to $(v,u)$ would not lower the indegree of the head of the edge $e$. An orientation is \emph{stable} if all edges are happy. See Figure~\ref{fig:ex-stable} for examples.

\begin{figure}[b]
\centering
\includegraphics[page=1]{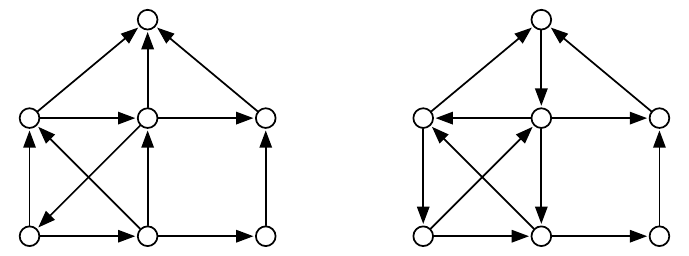}
\caption{Examples of stable orientations.
We can interpret each edge as a customer and each node as a server: a customer points to the server it is using, and the indegree of a server represents its load. A customer is happy if it cannot get a better service by unilaterally switching~servers.
}\label{fig:ex-stable}
\end{figure}

\paragraph{Customers and servers.}
We can interpret each edge as a \emph{customer} and each node as a \emph{server}.
If an edge $e = \{u,v\}$ is oriented from $u$ to $v$, then customer $e$ is using server $v$. The \emph{load} of a server is the total number of customers using it, i.e., its indegree. Customers would like to use servers with a low load, in order to maximize the quality of service they receive.

\paragraph{Centralized sequential algorithms.}

There is a simple centralized sequential algorithm that finds a stable orientation: start with an arbitrary orientation and then repeatedly pick an arbitrary  unhappy edge  and flip it. Flipping one edge may create new unhappy edges. However, it is easy to see that the algorithm will terminate in polynomial time in the number of nodes: the sum of squared indegrees is strictly decreasing.

This also shows that a stable configuration serves simultaneously two purposes: it is a \emph{game-theoretic equilibrium}, and it is also a \emph{local optimum} in a load-balancing problem in which the goal is to minimize the sum of squared loads.

\paragraph{Efficient distributed algorithms.}

The centralized algorithm is inherently sequential, and it may lead into a long propagation chain: flipping one unhappy edge creates another unhappy edge, flipping that one creates yet another, and such changes may eventually propagate throughout the graph.

Surprisingly, we can do much better in a distributed setting: \citet{Czygrinow2012} gave a distributed algorithm that finds a stable orientation in $O(\Delta^5)$ communication rounds, where $\Delta$ is the maximum degree of the graph. Remarkably, the running time is independent of the size of the graph, and only depends on $\Delta$. Even if we have an infinite graph, as long as the maximum degree is bounded, the nodes can collectively find a stable configuration in finite time. However, $O(\Delta^5)$ hardly sounds like a natural barrier for such a simple problem, but so far no improved algorithms or nontrivial lower bounds are known. In this work we present both.

The algorithm by \citet{Czygrinow2012} is highly efficient for graphs with very low degrees. In the case of large degree graphs, \citet{assadi2020improved} and \citet{halldorsson2018distributed} provide efficient approximation algorithms. Please refer to Section~\ref{sec: load balancing} for a more detailed discussion. In this work, we focus on the small degree case and do not try to optimize the dependency on $n$.

\begin{table*}
    \centering
    \begin{tabular}{@{}l@{\qquad}ll@{\qquad}ll@{\qquad}ll@{}}
    \toprule
    Problem
    & \multicolumn{2}{@{}l@{\qquad}}{Prior work}
    & \multicolumn{2}{@{}l@{\qquad}}{New upper bound}
    & \multicolumn{2}{@{}l}{New lower bound} \\ \midrule
    Token dropping
    &&
    & $O(L \cdot \Delta^2)$ & Thm.~\ref{thm: ballsruntime}
    & $\Omega(L + \Delta)$ & Thm.~\ref{thm: lowertoken}\\
    Token dropping, $2 \le L \le 3$
    &&
    & $O(\Delta)$ & Thm.~\ref{thm: 3layermagic}
    & $\Omega(\Delta)$ & Thm.~\ref{thm: lowertoken} \\
    Stable orientations
    & $O(\Delta^5)$ & \cite{Czygrinow2012}
    & $O(\Delta^4)$ & Thm.~\ref{theorem: stable}
    & $\Omega(\Delta)$ & Thm.~\ref{thm:orientation-lb}  \\ 
    Stable assignments
    & $O(S^5)$ & \cite{Czygrinow2012}
    & $O(C \cdot S^4)$ & Thm.~\ref{theorem: hyper}
    & $\Omega(S)$  & Thm.~\ref{thm:orientation-lb}   \\
    0-1-many assignments 
    &&
    & $O(C)$ & Thm.~\ref{thm: 2boundstable}
    & $\Omega(\min\{S,C\})$ & Thm.~\ref{thm: hyperlower012} \\
    \bottomrule
    \end{tabular}
    \caption{Summary of the new results and prior work: $\Delta$ = the maximum degree of the graph, $L$ = the height of the game, $C$ = the maximum degree of customers, and $S$ = the maximum degree of servers.}
    \label{table:summary}
\end{table*}

\subsection{Token dropping game}

\begin{figure}[t]
    \centering
    \includegraphics[page=2,width=0.8\columnwidth]{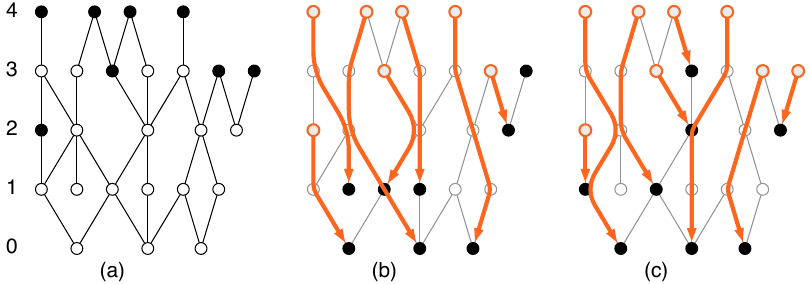}
    \caption{Token dropping game. (a)~Input; black nodes hold tokens. (b)--(c)~Examples of feasible solutions. The orange arrows represent the paths that the tokens followed until they reached their final positions.}\label{fig:ex-drop}
\end{figure}

The key new idea that we use to solve the orientation problem is to introduce a new graph problem that we call the \emph{token dropping game}. The game is illustrated in Figure~\ref{fig:ex-drop}. The input consists of a graph in which the nodes are organized in \emph{layers}, numbered from $0$ to $L$. Some of the nodes hold a token; a node can hold at most one token. The rules are simple:
\begin{framed}
    \noindent A token can move downwards from layer $\ell$ to layer $\ell-1$ along any edge to any node that does not currently hold a token. Each edge can be used at most once during the entire game.
\end{framed}
\noindent Put otherwise, once an edge has been used to move a token, it is deleted.
The task is to find a possible sequence of token movements such that we reach a configuration in which no token can be moved any further. I.e., the only goal of this single player game is to get stuck.

\paragraph{Centralized sequential algorithms.}

Again, there is a trivial centralized sequential algorithm for solving the token dropping problem: repeatedly pick any token that can be moved downwards and move it by one step. Eventually no token can be moved and the game is solved.

\paragraph{Efficient distributed algorithms.}

In this work we show that the token dropping problem can be solved in $O(L \cdot \Delta^2)$ rounds with a distributed algorithm; we also prove a lower bound of $\Omega(L + \Delta)$ rounds. Moreover, for $2 \le L \le 3$ we show that the problem can be solved in $O(\Delta)$ rounds, matching the lower bound.

\paragraph{Using token dropping to find stable orientations.}

We show that any algorithm that solves token dropping in $T(L,\Delta)$ rounds can be used to find a stable orientation in $O(\Delta \cdot T(\Delta,\Delta))$ rounds. Plugging in our algorithm for token dropping, we obtain an algorithm for finding a stable orientation in $O(\Delta^4)$ rounds, a factor-$\Delta$ improvement over the previous algorithm by \cite{Czygrinow2012}. We also prove a lower bound of $\Omega(\Delta)$ for any algorithm that finds a stable orientation.

\paragraph{New ideas.}

On a high level, the key new idea that enables us to save time in comparison with the prior algorithm is the following:
\begin{itemize}
    \item In the prior work, one starts with an arbitrary orientation. This potentially creates a large amount of unhappiness and resolving it takes a lot of time.
    \item In our work we orient edges more carefully, so that there is always at most one unit of excess load per node. We play token dropping with the excess load in order to resolve unhappiness. After $O(\Delta)$ such iterations, all edges are happily oriented.
\end{itemize}

\subsection{Generalization: stable assignment and semi-matchings}

Let us define a generalization of stable orientations as follows: Consider a bipartite graph in which we have customers on one side and servers on the other side. The task is to assign each customer to one server. Again, the customers would like to selfishly minimize the load of the server to which they are assigned to. We call a solution in which no customer wants to change its assigned server a \emph{stable assignment}.

Note that the stable orientation problem is a special case of the stable assignment problem with degree-$2$ customers only. The prior algorithm by \cite{Czygrinow2012} also solves the more general stable assignment problem. While our focus is on the orientation problem, we will also explain in this work how to generalize our algorithm beyond degree-$2$ customers. If the maximum degree of a customer is $C$ and the maximum degree of a server is $S$, our algorithm runs in $O(C \cdot S^4)$ rounds and the algorithm from prior work runs in $O(S^5)$ rounds. For the balanced case $\Delta = C = S$, both of the algorithms run in $O(\Delta^5)$ rounds. However, when the maximum degree of a customer is small, i.e., $C = o(S)$, our algorithm runs in $o(\Delta^5)$ rounds.

The stable assignment problem is closely connected to a load balancing problem known as \emph{semi-matching} \cite{Harvey2006}. There the task is to assign each customer to one server while minimizing the objective function $\sum_v f(g(v))$, where $g(v)$ is the number of customers assigned to server $v$ and $f(x) = 1 + 2 + \dotsb + x$. In essence, this is almost the same problem as minimizing the sum of squared loads. As observed by \cite{Czygrinow2012}, a stable assignment (in their terminology, a \emph{non-swappable semi-matching}) is also a factor-$2$ approximation of the optimal semi-matching. Hence our work gives a faster $2$-approximation algorithm for semi-matchings in the case of low-degree customers and high-degree servers.

A more refined analysis in \cite{Czygrinow2012} also shows that a stable assignment is a factor $1+1/\alpha$ approximation of the optimal semi-matching, where $\alpha = \max\{1,(\beta+1)/2\}$, and $\beta$ is the ratio of the number of servers to the number of customers. Hence when there are many customers and few servers, they have got a large $\beta$ and $\alpha$ and an approximation ratio close to $1$. Exactly the same holds for our algorithm, as this is a property of any stable assignment.

\paragraph{Open question for future work.}

A stable assignment can be used to find a $2$-approximation for semi-matching. However, any algorithm that finds a stable assignment takes at least $\Omega(\Delta)$ rounds. Is it possible to find a $2$-approximation of semi-matching in time $o(\Delta)$ by some other means (without going through a stable assignment)?

\subsection{Relaxation: 0--1--many assignments}

We conjecture that finding a stable orientation requires at least $\Omega(\Delta^2)$ rounds, i.e., it is strictly harder than e.g.\ the problem of finding a maximal matching in a bipartite graph, which is solvable in $\Theta(\Delta)$ rounds. If this is indeed the case, stable orientations would be a rare example of a natural graph problem that is solvable in $\poly(\Delta)$ but not in $O(\Delta)$ rounds (see Section~\ref{sec:related} for more discussion on related work).

We are currently still far from being able to prove superlinear lower bounds for stable orientations with the present techniques, but in this work we provide evidence suggesting that at least the general stable orientation problem is unlikely to be solvable in $O(\Delta)$ rounds.

To do this, we consider the following highly relaxed version of stable orientations: each customer is assigned to one adjacent server, and a customer does not want to use a server of load at least $2$ if there is a server of load $0$ available. In essence, this is a $0$--$1$--many version of stable orientations: customers only care about the difference between servers of load $0$, load $1$, and load at least $2$. We prove that the $0$--$1$--many version can be solved in $O(C)$ rounds, and in the balanced case $\Delta=S=C$, the problem requires $\Theta(\Delta)$ rounds. Hence the best upper bound for the relaxed version is much lower than for the general stable assignment problem, which is solvable in $O(\Delta^5)$ rounds. This suggests that the relaxed version is indeed strictly easier than the general version. If one could prove a strict separation between the relaxed version and the general version, then our $\Omega(\Delta)$ lower bound for the relaxed version would imply an $\omega(\Delta)$ lower bound for the general version, providing one of the first examples of a natural graph problem with a superlinear-in-$\Delta$ complexity.

\paragraph{Open question for future work.}

Can we show that the $0$--$1$--many version of stable assignments is strictly easier than the general version? For example, if we have an algorithm that solves the general version in graphs of maximum degree $\Delta$ in some time $T$, can we use it to solve the $0$--$1$--many version in graphs of maximum degree $\Delta' \gg \Delta$, in the same time~$T$?

\subsection{Organization of the paper}

Our paper is organized as follows. We start by discussing additional related work in \Cref{sec:related} and formalize the model of computing in \Cref{sec:model}. In \Cref{sec: token-dropping}, we introduce the token dropping game and give an upper and a lower bound for its complexity. In \Cref{sec: findStable}, we show how to use the token dropping to find stable orientation in $O(\Delta^4)$ rounds and in \Cref{sec:lowerorient}, we give a $\Omega(\Delta)$ lower bound. Then, in \Cref{sec: stable-ass}, we show to generalize our techniques to the stable assignment problem and discuss the complexity of the $0$-$1$-many relaxation of the stable assignment problem.

\section{Related work}\label{sec:related}

\subsection{Distributed complexity of locally verifiable problems}

This work is part of the ongoing effort of understanding the distributed computational complexity of \emph{locally verifiable problems}. In brief, these are problems in which a solution is globally correct if it looks good in all constant-radius neighborhoods. Stable orientations are by definition locally verifiable: if all edges are happy, the orientation is stable, and the happiness of an edge only depends on the other edges adjacent to it.

The study of locally verifiable problems in distributed computing plays a role similar to the study of the class NP in classical centralized sequential computing: given a problem in which solutions are easy to verify, what can we say about the complexity of \emph{finding} a feasible solution?

Typically, the complexity of locally verifiable problems is studied as a function of two parameters, the number of nodes $n$ and the maximum degree $\Delta$. In essence, these capture two complementary notions of scalability: how does the complexity of finding a solution increase when the input graph gets larger vs.\ when the input graph gets denser.

In general, the landscape of the distributed computational complexity for each possible combination of $n$ and $\Delta$ is complicated, but there are many problems that provide an opportunity to focus on one parameter only. To study the distributed complexity as a function of $n$, we can simply set $\Delta = O(1)$ and hence focus on bounded-degree graphs. In this case there are two important families of locally verifiable problems:
\begin{itemize}
    \item Symmetry-breaking problems, such as maximal matching, maximal independent set, vertex coloring, and edge coloring: all of these problems can be solved in $O(\log^* n)$ rounds \cite{cole86deterministic,Goldberg1988}, and this is tight \cite{Linial1992,Naor1991}.
    \item Orientation and splitting problems, such as sinkless orientation, sinkless and sourceless orientation, almost-balanced orientation, and almost-balanced splitting: all of these problems can be solved in $O(\log \log n)$ with randomized algorithms and in $O(\log n)$ rounds with deterministic algorithms \cite{ghaffari17distributed,Ghaffari2017a}, and these are tight \cite{Brandt2016,chang16exponential}.
\end{itemize}

The other dimension, dependency on $\Delta$, requires more care, as one cannot merely set $n = O(1)$ and study asymptotics as a function of $\Delta$. Therefore it is helpful to identify natural examples of graph problems that can be solved in $T(\Delta)$ rounds for some function $T$, \emph{independently of $n$}. In essence, we can set $n = \infty$ and study the complexity as a function of~$\Delta$. Key examples of problems that can be solved in $T(\Delta)$ rounds include:
\begin{itemize}
    \item Maximal matching on bipartite graphs can be solved in $O(\Delta)$ rounds \cite{Hanckowiak1998}, but not in $o(\Delta)$ rounds \cite{Balliu2019}.
    \item Maximal fractional matching can be solved in $O(\Delta)$ rounds \cite{Astrand2010}, but not in $o(\Delta)$ rounds \cite{Goos2017}.
    \item Weak coloring in odd-degree graphs can be solved in $O(\log^* \Delta)$ rounds \cite{Naor1995}, but not in $o(\log^* \Delta)$ rounds \cite{Brandt2019automatic}.
\end{itemize}
All of the above bounds are at most linear in $\Delta$. Stable orientation is perhaps one of the simplest locally verifiable graph problems that is known to be solvable in $T(\Delta)$ rounds, but for which the current upper bound is \emph{superlinear} in $\Delta$. By prior work, we do not have any nontrivial lower bounds for stable orientations, and the best upper bound is $O(\Delta^5)$. The recent advances in the techniques for proving lower bounds \cite{Brandt2019automatic,Olivetti2019,Balliu2019,Brandt2016} suggest that now would be a good time to revisit the stable orientation problem and see how far we can get in closing the gap between upper and lower bounds. In this work we take the first steps in this direction, by improving the upper bound to $O(\Delta^4)$ and by proving a lower bound of $\Omega(\Delta)$.

\subsection{Distributed load balancing}\label{sec: load balancing}

We point out that stable orientations can be interpreted as a distributed load balancing problem. Imagine that there is a load token on each edge; the task is to move each such token to one endpoint so that the load cannot be locally balanced any further.

Now if we let the tokens move freely further away from their original locations, we arrive at the \emph{locally optimal load balancing} problem, studied in \cite{Feuilloley2015}. This is a problem that can be solved in time $T(L,\Delta)$ for some $T$, where $L$ is the maximum initial load. However, it is an open question whether the problem can be solved in time $\poly(L,\Delta)$. It was conjectured that locally optimal load balancing cannot be solved in $\poly(L,\Delta)$ rounds, and if this is the case, stable orientations and token dropping are a strictly easier problems than locally optimal load balancing.

The key aspect that makes stable orientations and token dropping easier to solve than load balancing is the restriction that we can only use each edge once. If we ``move'' one unit of load over an edge by flipping the edge (in stable orientations) or by sliding a token along the edge (in token dropping), the edge cannot be used any more for moving additional load in the same direction. If there is a bottleneck that separates large high-load and low-load regions, an algorithm for load balancing has to essentially move load tokens across such an edge one by one until the load is locally balanced, while an algorithm for stable orientation or token dropping will use the edge only once.

\citet{halldorsson2018distributed} considered the backup placement problem, where clients are tasked to select a set of $k$ neighboring servers on which backups are placed. They provide approximation algorithms for different optimization goals, e.g., minimize the maximum load while satisfying all clients, or maximize the number of satisfied clients. They give a randomized \congest algorithm that provides an approximation ratio of $O( \log n  / \log \log n)$ in $\polylog(n)$ rounds. In a similar vein, \citet{halldorsson2019distributed} gave randomized $\polylog(n)$-time approximation algorithms for the $k$-server assignment problem, where servers also have maximum capacity, and the objective is to maximize the total profit of satisfied clients subject to the server capacities.

\citet{assadi2020improved} gave randomized approximation algorithms for unweighted and weighted load balancing. For the \congest model, they gave algorithms that achieve a $O(1)$-approximation for unweighted graphs and $O(\log n)$-approximation for weighted graphs in $\polylog(n)$ rounds. Moreover, in the \local model, it is possible to achieve a $O(1)$-approximation algorithm in $\polylog(n)$ rounds with high probability. Recently, \citet{ahmadian2021convex} considered load balancing under general convex objective functions, where fractional solutions are permitted, and gave algorithms that give near-optimal solutions in $\log n \cdot \polylog(\Delta)$ rounds in the \congest model.

In this work, we focus on deterministic algorithms whose running times are independent of $n$. We operate in the \local model, but the message complexity of our algorithms are small, as the algorithms are based on only simple proposal strategies, where nodes request for single unit of load to be transferred. Thus, the algorithms can also be run in the \congest model, but also our lower bounds hold in the stronger \local model.

\section{Preliminaries}\label{sec:model}

In this work, we consider the standard \local model of distributed computing introduced by \cite{Linial1992}: Each node of the input graph $G = (V, E)$ is a computational entity and each edge $e = \{u, v\} \in E$ represents a bidirectional communication link. Computation proceeds in synchronous communication rounds and the message sizes are unbounded. The nodes are equipped with unique identifiers and initially, the only information that a node $u$ has are the identifiers of its neighbors. Throughout the paper, $n$ denotes the number of nodes and $\Delta$ denotes the maximum degree of the input graph.
We emphasize that even though we discuss directed edges in this work, communication is always allowed in both directions over a communication~link.

\section{The token dropping game}\label{sec: token-dropping}

In this section, we introduce the token dropping game slightly more formally and present our algorithm for solving the game. In the end of the section, we complement this result with a lower bound of $\Omega{(\Delta)}$ communication rounds via a reduction to the maximal matching problem. Interestingly, this lower bound already holds for games with $2$ levels.

The input for the token dropping game consists of a directed graph $G = (V, E)$ that contains no directed cycles and a set of \emph{tokens} $S = s_1, \ldots, s_k$. The tokens are assigned to the nodes such that each node contains at most one token. Furthermore, each node $v$ is assigned a level $\ell(v) \leq L$, where $L$ denotes the \emph{height} of the game. The levels of the nodes and the assignment of the tokens are given by an adversary. The nodes are not aware of any parameters, such as their level, the maximum degree $\Delta$, or the number of nodes $n$ in the beginning of the execution.

We say that if there is a directed edge $(u, v) \in E$ from $u$ to $v$, then $v$ is a \emph{parent} of $u$ and conversely, $u$ is called a \emph{child} of $v$. If $v$ is a parent of $u$, then the level function must satisfy the condition\footnote{All of our results work even if we allow that $\ell(v) >\ell(u) + 1$ for a parent $v$ of $u$. For the sake of the presentation, we chose to restrict the discussions to the case where the edges are between adjacent levels.} $\ell(v) = \ell(u) + 1$.

\paragraph{Objective.}

The basic principle is that node $u$ can pass a token to a child and in the end of the execution, $u$ is only allowed to possess a token if it cannot pass its token to any of its children. For each token $s$, the goal is to find a path $p_s = (v_1, \ldots, v_d)$ from its original node $v_1$ to its destination $v_d$, where for every $1 < i \leq d$, node $v_{i - 1}$ is the parent of node $v_{i}$. Formally, for each node $v$, the output is a set of pairs of incoming and outgoing edges, where each pair corresponds to a path of a token traveling through.\footnote{Note that the traversals of the tokens can be derived from the node-centered output in at most $L$ communication rounds.} If $v$ initially contains a token, then the set is allowed to have one singleton outgoing edge in the set and similarly, if $v$ is the destination of a token, there can be one singleton incoming edge. Notice that the token dropping game satisfies the preconditions of a locally checkable problem.

The path $p_s$ for token $s$ is referred to as the \emph{traversal} of $s$. There are three rules:
\begin{enumerate}
    \item Each edge is used at most once, i.e., the traversals are edge-disjoint. We say that an edge is \emph{consumed} once it is traversed by a token.
    \item The destination node for each token traversal is unique, i.e., for any two traversals $p_{s_1} = (v_1, \ldots, v_d)$ and $p_{s_2} = (u_1, \ldots, u_{d'})$ it holds that $v_d \neq u_{d'}$.
    \item Each traversal is \emph{maximal}, i.e., if $v$ is the destination node of a traversal $p_s$, then each incoming edge $(u, v)$ is either consumed by another traversal or child $u$ is the destination of another traversal.
\end{enumerate}

\subsection{The proposal algorithm}\label{sec: proposal}

Now, we present our algorithm for the token dropping game. The algorithm follows a simple proposal strategy, where a token is passed to a child whenever that is possible without causing any conflicts with other tokens that are passed. One of the most important ingredients of our analysis is to identify a way to measure the progress of a token on its traversal. In \Cref{lemma: activerounds}, we show that if a node is making proposals, then many edges are being consumed in its neighborhood. Then, in \Cref{lemma: activenode}, we show that, for any token $s$, we can find a fixed directed path (see \Cref{def: extended}) of nodes such that, if $s$ has not yet reached its destination, at least one node on this path is making progress.
Once all edges in the $2$-hop neighborhood of this path are consumed, the token $s$ must have reached its destination.
The goal of the rest of the section is to provide a proof for the following theorem.

\begin{theorem}\label{thm: ballsruntime}
    There is an algorithm that solves the token dropping game in $O(L \cdot \Delta^2)$ rounds, where $L$ is the height of the game.
\end{theorem}

\paragraph{Algorithm details.}

We call a node \emph{active} if at least one of its parents has a token. Furthermore, we call a node \emph{occupied} if it contains a token and \emph{unoccupied} otherwise. Our algorithm works as follows. In every round, every active and unoccupied node requests a token from some parent that has a token, ties broken arbitrarily. If a node receives at least one request, then it passes the token to one (arbitrarily chosen one) of the children it received a request from. Notice that upon passing the token, the edge to the corresponding child is consumed and hence, removed from the game. If a node $u$ is occupied and has no children or is unoccupied and has no parents, then $u$ terminates. When a node terminates, we also remove it from the game. We note that each round of our algorithm actually consists of two synchronous communication rounds but for the sake of the presentation, we combine two communication rounds into one round for the rest of the discussion.

\begin{lemma}\label{lemma: correctness}
    The output of the proposal algorithm is correct.
\end{lemma}
\begin{proof}
    It is easy to verify that the traversals are edge-disjoint. 
    Upon traversal an edge is consumed, and hence any edge is traversed by at most one token.
    For maximality, suppose for a contradiction that there is an unoccupied node $u$ that has a parent $v$ with a token and edge $e = (u, v)$ was not consumed.
    Since $e$ is not consumed and $u$ is unoccupied, $v$ must have terminated before $u$.
    However, that is a contradiction since an occupied node does not terminate if it has any children.  
\end{proof}

\begin{figure}[t]
    \centering
    \includegraphics[page=3,width=0.8\columnwidth]{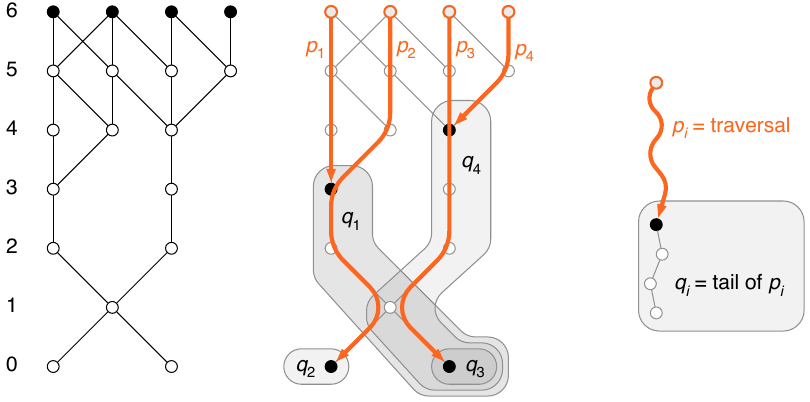}
    \caption{Examples of traversals $p_i$ and their tails $q_i$ (represented by the gray areas). We assume that the token following path $p_2$ reached its final destination before the token following path $p_3$.}\label{fig:tail}
    \label{fig: tails}
\end{figure}

\begin{definition}\label{def: extended}
    Consider the traversal $p_s = (v_1, \ldots, v_d)$ of token $s$ given by the proposal algorithm. We define the \emph{tail} of traversal $p_s$ as the longest path $(v_{d}, \ldots, v_{h})$ starting in $v_d$ with the property that, for any $d \leq i \leq h-1$, node $v_i$ passes at least one token to a child, and the last token $v_i$ passes down goes to node $v_{i+1}$. If $v_d$ did not pass any tokens further down, the tail only contains the node $v_{d}$. We refer to the concatenation $p^*_s = (v_1, \ldots, v_h)$ as the \emph{extended} traversal of $s$. See \Cref{fig: tails} for an illustration.
\end{definition}

\begin{lemma} \label{lemma: activerounds}
    Any node $u$ can be active and unoccupied in at most $O(\Delta^2)$ rounds.
\end{lemma}
\begin{proof}
    Consider a round where $u$ is active and unoccupied.
    By definition of being active, there is at least one parent of $u$ that has a token and hence, $u$ requests a token from some parent $v$ of $u$.
    The parent $v$ will accept exactly one proposal that it receives and hence, its token will be passed on to a child $u'$ and the corresponding edge $(u', v)$ will be consumed.
    In other words, in every round that $u$ is active and unoccupied, at least one edge incident to some parent of $u$ will be consumed.
    Since, there are at most $\Delta$ parents of $u$ that have degree at most $\Delta$ each, all of their edges are consumed after $\Delta^2$ rounds in which $u$ is active and unoccupied. 
\end{proof}

\begin{lemma} \label{lemma: activenode}
    Consider an arbitrary token $s$ with traversal $p_s = (v_1, \ldots, v_d)$. At any point in time $t$ at which $s$ has not reached $v_d$ yet, at least one node is active and unoccupied in the extended traversal $p^*_s = (v_1, \ldots, v_d, \ldots, v_h)$ of $s$.
\end{lemma}
\begin{proof}
    Let $v_i$, for some $1 \leq i < d$ be the node that holds token $s$ at time $t$.
    Notice that by definition, the first unoccupied node $v_j$, with $i < j \leq h$ along $p^*_s$ has a parent with a token and hence, is active.
    Suppose for a contradiction that all nodes on the tail of traversal $p_s$ are occupied.
    By the definition of a tail, since the node $v_{h}$ is the last node in the tail, it will never pass its token to any of its children.
    Therefore, it will never become unoccupied.

    Again, by the definition of a tail, node $v_{h}$ is the last node that node $v_{h - 1}$ passes a token to.
    Since $v_{h}$ is already occupied and will never pass its token, it must be the case that it is already holding the token that $v_{h - 1}$ passes as its last token.
    Therefore, $v_{h - 1}$ will never become unoccupied.
    By induction, this holds for all the nodes on the tail of $p_s$, including the destination node $v_d$ of $s$.
    This contradicts the fact that node $v_d$ is the destination of $s$ and hence, yields the lemma.
\end{proof}

\begin{proof}[Proof of \Cref{thm: ballsruntime}]
    Consider an arbitrary token $s$. By definition, the length of the extended traversal $p^*_s$ of $s$ is at most $L$.
    It easy to verify that once all tokens have reached their destinations, all nodes will terminate in $O(L)$ rounds.
    The theorem follows by combining \Cref{lemma: activerounds,lemma: activenode,lemma: correctness}.
\end{proof}

\subsection{A lower bound}

To complement our upper bound for the token dropping game, we show a reduction from the maximal matching problem to the token dropping game.
A detail that might be of independent interest is that the lower bound already holds for the case of only $2$ levels.
We were able to design an algorithm for the case of at most $3$ levels that matches (the $\Delta$ dependency of) this bound but our approach does not seem to generalize beyond $3$ levels; the complexity of the token dropping game for $4$ or more levels is left as an open question for future work.

\begin{theorem}\label{thm: lowertoken}
    There is no deterministic algorithm that solves the token dropping game in $o(\Delta + \log n / \log \log n)$ rounds in the \local model. This holds even when restricted to games of height $2$.
\end{theorem}
\begin{proof}
    The proof of the theorem is a simple reduction from the bipartite maximal matching problem to the token dropping game. In a recent work, \cite{Balliu2019} showed that the bipartite maximal matching problem cannot be solved in $o(\Delta + \log n / \log \log n)$ rounds in the \local model of distributed computing. Consider a bipartite graph $G = (U\cup V, E)$ that is an input instance to the maximal matching problem. We create a token dropping instance with $2$ levels by considering every node $u \in U$ as a level-$1$ node with a token and every node $v \in V$ as a level-$0$ node. The traversals of the output for the token dropping game directly correspond to a maximal matching completing the reduction.
\end{proof}

\subsection{Token dropping through 3 levels}\label{sec: 3levels}

In this section, we show an algorithm that runs in $O(\Delta)$ rounds and solves the token dropping game when the number of levels is bounded by $3$. Our approach, however, does not seem to generalize for a larger number of levels. For the case of $4$ levels, the current best algorithm has runtime of $O(\Delta^2)$ (from \Cref{thm: ballsruntime}) introducing a gap of factor $\Delta$ between the cases of $3$ and $4$ levels. In the hope of finding better algorithms for an arbitrary number of levels, we believe that it is an interesting first step to solve the case of $4$ levels in time $o(\Delta^2)$.

\paragraph{Our algorithm for 3 levels.}

Our approach is a refined version of the proposal algorithm from \Cref{sec: proposal}. In the case of $3$ levels, we can leverage the fact the highest and the lowest layers only have level-$1$ nodes as neighbors. Inspired by this, the nodes in level $1$ take an active role and handle moving the tokens. More precisely, in every round, each active and unoccupied node in level $1$ requests a token from a parent that contains a token. Each node in level $2$ that gets a request passes its token to one child that made a request. Furthermore, each occupied node in level $1$ makes a proposal to an unoccupied child. Then, each node in level $0$ that receives a proposal accepts one of the received proposals and thereby the offered token.

Nodes in level $2$ terminate as soon as they are unoccupied and get removed from the game. Level $0$ nodes terminate if they are occupied or have no more parents left. Finally, nodes in level $1$ terminate if they are unoccupied and have no parents or if they are occupied and have no children.

\begin{theorem}\label{thm: 3layermagic}
    There is a deterministic algorithm that solves the token dropping game of height $3$ in $O(\Delta)$ rounds in the \local model.
\end{theorem}
\begin{proof}
    Due to the design of our protocol for $3$ levels, it is clear that no node will ever have $2$ tokens and all tokens are eventually passed down if there is an unoccupied child. For the runtime analysis, consider some node $u$ in layer $1$. Recall that if $u$ is occupied and has no children or if $u$ is unoccupied and has no parents, then $u$ terminates. Hence, we can assume that in every round, node $u$ either makes a request to a parent or a proposal to a child. In the case that $u$ requests a token from some parent $v$, the parent $v$ will accept at least one proposal. In this case, node $v$ will pass its token and become unoccupied and since $v$ is in level $2$ it will terminate. In the case that $u$ proposes to some child $c$, this child will accept at least one proposal. Node $c$ will receive a token and become unoccupied and since $c$ is in level $0$ it will terminate. Therefore, in every round, at least one neighbor of $u$ will terminate and hence, the runtime bound of $O(\Delta)$ communication rounds follows.
\end{proof}

\section{Finding a stable orientation}\label{sec: findStable}

In this section, we show how to efficiently find a stable orientation. The key idea is to utilize our algorithm for the token dropping game as a black box to maintain a stable partial solution throughout the execution and to carefully and gradually extend the partial stable solution to a complete stable solution.

\begin{theorem}\label{theorem: stable}
    There is a deterministic algorithm that finds a stable orientation in $O(\Delta^4)$ communication rounds. This runtime is independent of the size of the input graph.
\end{theorem}

\subsection{An overview of our algorithm}

The basic idea behind our algorithm is to start with an unoriented graph and gradually orient the edges until all edges are oriented. We split the execution of our algorithm into \emph{phases} and our goal is to guarantee that in the end of each phase, there are no unhappy directed edges in the graph. Let $G = (V, E)$ be the input graph that is initially unoriented. In the beginning of a phase, each unoriented edge sends a proposal to its endpoint with the smaller load\footnote{For the sake of presentation, it is convenient to think of the edge as the actor for sending a proposal. However, this proposal is easy to implement in the node-centered view as well.}, breaking ties arbitrarily. Each node $u$, that receives a proposal, will accept exactly one of the proposals. When a proposal is accepted by node $u$, we will orient the corresponding edge towards $u$---however, before doing so, we make a preparation step on the graph induced by the already oriented edges in order to avoid creating any unhappy edges due to the new orientations. This preparation step is where we apply the token dropping game as a black box.

\subsection{Utilizing the token dropping game}

The \emph{badness} of a directed edge $(u, v)$ is defined as $\indegree(v) - \indegree(u)$. An important observation is that as soon as the badness of an edge $(u, v)$ is strictly larger than $1$, then by flipping the edge, the badness of the edge is reduced and the edge becomes happy. Furthermore, if the badness is at most $1$, then the edge is happy. Let us suppose that every directed edge is happy in the beginning of a phase. Then, we know that the maximum badness for any edge in the graph is at most $1$. We create a token dropping instance by including all directed edges that have badness exactly $1$. All nodes are added into the token dropping instance, even if they end up isolated in the token dropping game. The nodes are assigned to levels according to their current load. In addition, for each unoriented edge selected to be oriented towards node $u$ in the current phase, we add a token to node $u$. We note that these unoriented edges are not included in the token dropping instance (in this phase).

Then, we run the token dropping algorithm on the instance we created and obtain a set of traversals. We re-orient all the edges according to the traversals or, in other words, flip every edge present in the traversals. We show in \Cref{lemma: badness} that after flipping the edges, we have badness bounded by $1$ and hence, we have our invariant that allows us to proceed to the next phase. In \Cref{lemma: gameheight} we show that the created token dropping instance is valid and has height bounded by $\Delta$. Finally, in \Cref{lemma: numPhases}, we give a bound on the number of phases we require and are ready to prove \Cref{theorem: stable}.

\begin{lemma}\label{lemma: gameheight}
    The created token dropping instance is valid and the height of the game is at most~$\Delta$.
\end{lemma}
\begin{proof}
    Since the nodes are assigned to levels according to their loads, the bound for the height of the game follows from the fact that the maximum load of a node is bounded by the maximum degree. Every node can have at most one token since in each phase each node accepts at most one proposal. Finally, all edges have badness exactly one, which implies that they go from a node in layer $i$ to some other node in layer $i +1$.
\end{proof}

\begin{lemma}\label{lemma: matchCorrect}
    Consider a node $v$ and some phase $p$. The load of $v$ increases by $1$ in phase $p$ if and only if $v$ is the destination of a token in the token dropping game created in phase $p$. Otherwise, the load of $v$ does not change in phase $p$.
\end{lemma}
\begin{proof}
    According to the design of our protocol, the orientations of all the edges contained in the traversals are flipped and hence, flipping the edges will not affect the load of nodes that are not the endpoints of a traversal. Consider now the case that $v$ is a starting point of a traversal but not an endpoint of any traversal. Flipping the edges decreases the load of $v$ by $1$, but directing the undirected edge in the end of the phase cancels out the decrease. If $v$ is the endpoint of a traversal, we need to consider three cases. First, if the traversal corresponds to the token staying still, then no edges are flipped and orienting the undirected edge will bring an increase of $1$ in the load. Second, if $v$ is an endpoint of a traversal and also a starting point of one (and we are not in the previous case), then the flipped edges cancel each other out and the orienting of the undirected edge will increase the load by $1$. Finally, if $v$ is not the starting point but is an endpoint of a traversal, then turning the edges will increase the load by one. That covers all cases, yielding the lemma.
\end{proof}

\begin{lemma}\label{lemma: badness}
    In the end of a phase, there are no directed edges with badness larger than $1$.
\end{lemma}
\begin{proof}
    We approach the proof by induction. For the first phase, the claim follows by observing that for each node $v$, at most one edge is directed towards $v$. Consider now some phase $p$ where the badness of each directed edge is initially bounded by $1$. Suppose for a contradiction that in the end of the phase, there is an edge $e = (u, v)$ with badness at least $2$. First, consider the case that $e$ was unoriented in the beginning of phase $p$. Since $e$ proposes the endpoint with the smaller load, it must be the case that $v$ has the smaller load in the beginning of phase $p$. According to \Cref{lemma: matchCorrect}, the load of a node cannot decrease and can increase by at most $1$ yielding a contradiction.

    Then, consider the case that edge $e$ was oriented in some previous phase for the first time. According to the induction assumption, the badness of edge $e$ was at most $1$ in the beginning of phase $p$. According to \Cref{lemma: matchCorrect}, the load of a node can only increase by one per phase and hence, it must be the case that $e$ was oriented from $u$ to $v$ in the beginning of phase $p$. Furthermore, since the badness increased from $1$ to $2$, we have that node $v$ was an endpoint of a traversal and that $u$ was unoccupied in the end of the token dropping game in phase $p$. Also, $(u, v)$ was not traversed, because it is still oriented towards $v$ in the end of phase $p$ (and would have been flipped otherwise). However, due to the design of our protocol, edge $(u, v)$ was a part of the token dropping game and hence, node $u$ not having a token violates the maximality of the token dropping game. This completes the inductive step and the claim follows by induction.
\end{proof}

\begin{lemma}\label{lemma: numPhases}
    The number of phases is $O(\Delta)$.
\end{lemma}
\begin{proof}
    Consider an arbitrary undirected edge $e = \{u, v\}$. In every phase, $e$ sends a proposal to either $u$ or $v$. Suppose w.l.o.g., that $u$ receives the proposal. Since $u$ accepts at least one proposal it receives, we have that at least one unoriented edge incident on $u$ becomes oriented. Hence, after $O(\Delta)$ phases, edge $e$ has to be oriented since in each phase before $e$ becomes oriented, at least one of the $2 \Delta - 2$ edges incident to $u$ or $v$ and different from $e$ has to change from unoriented to oriented.
\end{proof}

\begin{proof}[Proof of \Cref{theorem: stable}]
    By \Cref{lemma: gameheight}, the token dropping instances we create in every phase are of height at most $\Delta$. Hence, by \Cref{thm: ballsruntime}, we get that it takes $O(\Delta^3)$ rounds to execute one phase of our algorithm. Combining this with \Cref{lemma: numPhases}, we get a runtime bound of $O(\Delta^4)$ communication rounds. The correctness of the algorithm is given by \Cref{lemma: badness}.
\end{proof}

\section{Linear-in-\texorpdfstring{$\boldsymbol\Delta$}{} lower bound for stable orientations}\label{sec:lowerorient}

In this section, we show that finding a stable orientation takes $\Omega(\Delta)$ rounds.
We define that a perfect $d$-regular tree of depth $k$ is a rooted tree, where (1) every non-leaf node has degree $d$ and (2) every leaf node is at distance $k$ from the root node. The \emph{height} $h(v)$ of a node $v$ is its distance to the closest leaf node; if $v$ is a leaf, then $h(v) = 0$.

\begin{lemma}\label{lemma:stable-height}
    Let $G = (V,E)$ be a perfect $d$-regular tree. In any stable orientation, $\indegree(v) \le h(v) + 1$.
\end{lemma}
\begin{proof}
    Let $\mathcal{C}(v)=\{ u : h(u) = h(v)-1 \text{ and } (u,v) \in E\}$ denote the children of $v$. Consider an arbitrary stable orientation of $G$. We show by induction on the height that $\indegree(u) \le h(v)$ for $u \in \mathcal{C}(v)$. The claim follows from this. Consider nodes at height $i$. The base case $i=0$ is trivial, as leaves have no children. Let $v$ be a node with $h(v)=i$. Suppose $\indegree(v)\ge i+2$. As $v$ has only one parent, at least one child $u\in\mathcal{C}(v)$ must have its edge pointed at $v$. By the induction assumption we know that for every child $u\in\mathcal{C}(v)$ it holds that $\indegree(u)\le i$. This implies that the edge $e=(u,v)$ is unhappy, and thus, the orientation of $G$ is not stable.
\end{proof}

\begin{lemma}\label{lemma:heavy-nodes}
    Let $G = (V,E)$ be an oriented $d$-regular graph. Then there exists a node $v \in V$ such that $\indegree(v) \ge \lceil d/2 \rceil$.
\end{lemma}
\begin{proof}
    For the sake of contradiction, suppose the claim does not hold. The sum over all $n$ nodes will then yield a result strictly smaller than $nd/2$. However, in any $d$-regular graph the number of edges is given by $|E|=nd/2$. Thus, we obtain a contradiction as
    \[
    nd/2 = |E| = \sum_{v \in V} \indegree(v) < nd/2.\qedhere
    \]
\end{proof}

\begin{theorem}\label{thm:orientation-lb}
    Any algorithm that finds a stable orientation has a running time of $\Omega(\Delta)$ rounds.
\end{theorem}

\begin{proof}
    Fix $\Delta$ and suppose there exists an algorithm $\mathcal{A}$ that outputs a stable orientation in $t \le \Delta/2-3$ rounds. Fix a $\Delta$-regular graph $G_1 = (V_1,E_1)$ with girth at least $\Delta+1$; for sufficiently large $n$ such graphs exist. Consider the orientation produced by $\mathcal{A}$ in $G_1$. By Lemma~\ref{lemma:heavy-nodes}, there exists some $v$ that has $\indegree(v) \ge \lceil \Delta/2\rceil$ in this orientation.

    Next, let $G_2 = (V_2,E_2)$ be a perfect $\Delta$-regular tree of depth $\Delta+1$. Pick a node $v' \in V_2$ such that $h(v') =\lceil \Delta/2\rceil-2$. Let $G[v,t]$ denote the $t$-radius neighborhood of node $v$ in graph $G$. Clearly, $G_1[v,t]$ and $G_2[v',t]$ are isomorphic, as the $t$-radius neighborhoods of $v$ and $v'$ are $\Delta$-regular trees and indistinguishable. Hence, $\mathcal{A}$ produces the same output for $v$ and $v'$, i.e., $\indegree(v) = \indegree(v')$. By Lemma~\ref{lemma:stable-height}, any orientation output by $\mathcal{A}$ in $G_2$ satisfies $\indegree(v') \le h(v')+1$. Thus, we have that
    \[
    \lceil\Delta/2\rceil \le \indegree(v) = \indegree(v') \le h(v') +1 = \lceil\Delta/2\rceil-1,
    \]
    which is a contradiction, thus yielding the claim.
\end{proof}

\section{Stable assignments}\label{sec: stable-ass}

In this section, we study the stable assignment problem. Recall that in this problem, we have customers on one side of a bipartition and servers on the other and the task of the customers is to choose exactly one server such that the load of the server is minimized. The selection of servers is stable if no customer has an incentive to change their choice. Also, recall that the stable orientation problem is a special case of the stable assignment problem, where each customer has degree $2$. Throughout the section, we denote the maximum customer degree by $C$ and the maximum server degree by $S$, and we use $\Delta = \max \{C,S\}$ to denote the maximum degree in the entire network. We give two results on the stable assignment problem.

First, in \Cref{sec: hypertoken,sec: hyperassignment}, we show that the proposal algorithm and the scheme of gradually orienting edges are robust to higher customer degrees. We interpret the bipartite input graph as a hypergraph, where the customers act as hyperedges. We define the token dropping game on hypergraphs and explain how to adapt the arguments from the case of rank $2$ customers to solve the hypergraph version. Then, we show how to gradually orient hyperedges such that the badness of any hyperedge is at most $1$ in the end of each phase.

Second, in \Cref{sec: hyperlower}, we will study a relaxation of the stable assignment problem. We consider the variant where all loads above a certain threshold are considered equal, and we will show that already with very small thresholds, this problem is at least as hard as maximal matching. Furthermore, we give an algorithm with a strictly faster runtime than what we obtained for the general version.

\subsection{Token dropping for stable assignment}\label{sec: hypertoken}

An oriented hyperedge is an edge where one node has the special role of the \emph{head} of the edge. The other nodes in the hyperedge are oriented towards this edge, i.e., serve the role of the tail of the hyperedge. We generalize the token dropping game by adapting all definitions and rules in the natural way. In particular:

For each hyperedge $e = \{ v_1, \dots, v_i \}$ with head $v_1$, we have \[\ell(v_1) = \min\{ \ell(v_2), \dots, \ell(v_i) \} + 1.\] For two endpoints $u, v$ of a hyperedge $e$, we say that $u$ is a parent of $v$  and $v$ a child of $u$ (\emph{in hyperedge $e$}) if $u$ is the head of $e$, and $\ell(u) = \ell(v) + 1$. A token can only be passed by the head of some hyperedge to one of its children in the hyperedge---analogously to before, this process includes that the hyperedge is consumed. The three rules of (hyper)edge-disjoint traversals, unique destinations, and maximal traversals hold analogously.

\paragraph{The proposal algorithm.}

Similarly to the case of rank $2$, unoccupied nodes propose to a parent with a token, and occupied nodes pass a token to a child that made a proposal (to the node). The proofs of \Cref{lemma: activenode,lemma: correctness,lemma: activerounds} can be adapted in a straight-forward manner. To see why a node can be active for at most $O(S^2)$ rounds, as promised in \Cref{lemma: activerounds}, one needs to observe that each hyperedge has only one head and that the whole hyperedge is consumed when a token is passed through it. Hence, each parent (of which there are at most $S$) needs to be proposed to at most $S$ times. By the adapted lemmas, we obtain the same result for the hypergraph setting as for the case of customers of degree $2$.

\begin{theorem}\label{thm: hyperballsruntime}
    There is an algorithm that solves the token dropping game in $O(L \cdot S^2)$ rounds, where $L$ is the height of the game.
\end{theorem}

\subsection{Finding a stable assignment}\label{sec: hyperassignment}

\paragraph{Gradually orienting edges.}

Similarly to the case of rank $2$, our plan is to divide the execution of our algorithm into phases and to guarantee a maximum badness of at most $1$ at the end of every phase, where the badness of a hyperedge $e = \{ v_1, \dots, v_i \}$ with head $v_1$ is defined as $\indegree(v_1) - \min\{ \indegree(v_2), \dots, \indegree(v_i) \}$. In each phase, every unoriented hyperedge makes a proposal to the node with the smallest load and exactly one proposal is accepted by any node that received at least one proposal. Then the algorithm proceeds as described in \Cref{sec: findStable}, where flipping an edge now corresponds to changing the head of a hyperedge: if, in the token dropping game, a token was passed from node $u$ to node $v$ via hyperedge $e$, then the head of $e$ changes from $u$ to $v$.

As before, the token dropping instance is created from hyperedges with badness exactly $1$ and all nodes and the tokens are added to the nodes that accepted a proposal. Now, all the statements from \Cref{sec: findStable} can be generalized in a straightforward manner. For the generalization of \Cref{lemma: numPhases}, we obtain a slightly worse bound than for \Cref{lemma: numPhases}, as shown in \Cref{lemma: hyperphases}. Then, \Cref{theorem: hyper} follows from \Cref{thm: hyperballsruntime} and \Cref{lemma: hyperphases}.

\begin{lemma}\label{lemma: hyperphases}
    The number of phases is at most $O(C \cdot S)$.
\end{lemma}
\begin{proof}
    Consider an arbitrary hyperedge $e \subseteq V$ of a hypergraph $H = (V, E)$. In every phase, if $e$ is not yet oriented, $e$ sends a proposal to one of its nodes. This node must accept a proposal from at least one of its nodes. Since the rank of $e$ is at most $C$ and the maximum degree of a node is at most $S$, there can be at most $C \cdot S$ phases until $e$ is oriented itself, or all edges incident to the nodes of $e$ are oriented, which implies that $e$ becomes oriented in the next phase at the latest.
\end{proof}

\begin{theorem}\label{theorem: hyper}
    There is an algorithm that solves the stable assignment problem in $O(C \cdot S^4)$ communication rounds in the \local model.
\end{theorem}

\subsection{Relaxations of the stable assignment problem}\label{sec: hyperlower}

An interesting variant of the stable assignment problem is obtained by considering all possible loads above a certain threshold as being the same: after all, if the server has too high of a load, the customer might not care anymore how high the load is exactly, and instead looks for another solution. From a theoretical point of view, we might hope to solve these relaxations faster than the original stable assignment problem; in particular, the linear-in-$\Delta$ lower bound presented in Section \ref{sec:lowerorient} weakens proportionally to the chosen threshold until we remain with a constant lower bound for relaxations with a constant threshold. While it is an intriguing open question how much the stable assignment problem becomes easier by introducing such a threshold, we show in \Cref{thm: hyperlower012} that even in the most relaxed non-trivial case, i.e., if we consider all loads strictly above $1$ as equal, we cannot hope for a better than linear dependency on $\Delta$. Formally, for each $k \geq 2$, the \emph{$k$-bounded stable assignment problem} is defined as the original stable assignment problem with the only difference that customers are only unhappy if they have chosen a server with indegree $\ell$, but also have a neighbor of load at most $\min\{ k, \ell \} - 2$. In particular, in the $2$-bounded stable assignment problem, the only unhappy customers are those that have a neighbor with indegree $0$, but have chosen a server with indegree at least $2$.

\begin{theorem}\label{thm: hyperlower012}
    There is no deterministic algorithm that solves the $2$-bounded stable assignment problem in $o(\Delta + \log n / \log \log n)$ rounds in the \local model.
\end{theorem}
\begin{proof}
    Similarly to the proof of Theorem \ref{thm: lowertoken}, we use a reduction from the bipartite maximal matching problem to the $2$-bounded stable assignment problem in order to obtain the desired lower bound. Consider a bipartite graph $G = (U\cup V, E)$ that is an input instance to the maximal matching problem and let $S = C = \Delta$. Now, first we find a solution to the $2$-bounded stable assignment problem on $G$, where the nodes in $U$ are the customers, the nodes in $V$ are the servers, and we interpret the edges which connect a customer with the chosen server as our (preliminary set of) matching edges. Then, each server with more than one incident edge in the preliminary matching keeps exactly one of those edges as a matching edge, and removes all of the others from the preliminary matching. In the following we show that the resulting edge set is a correct solution to the maximal matching problem.

    The correctness of the output to the $2$-bounded stable assignment problem and the post-processing step ensure that each node is matched to at most one other node. Now consider an unmatched customer node $u$. Node $u$ is unmatched because the chosen server $v$ removed the connecting edge $\{ u, v \}$ from the matching, for which $v$ must have had another incident edge in the preliminary matching. This implies that in the solution to the $2$-bounded stable assignment problem, $v$ must have had load at least $2$, which in turn implies that $u$ has no neighbor of load $0$ in that solution. Hence, every neighbor of $u$ is matched in the preliminary matching, and hence also in the final matching (which allows $u$ to be unmatched).

    Finally, consider an unmatched server node $w$. Node $w$ must have had load $0$ in the solution to the $2$-bounded stable assignment problem, and therefore each neighbor of $w$ must have chosen a server of load at most $1$, according to the definition of the $2$-bounded stable assignment problem. This implies that for each neighbor of $w$, the incident edge in the preliminary matching must also be in the final matching as load-$1$ servers do not remove any incident edge from the preliminary matching. Hence, $w$ being unmatched does not violate the maximal matching constraints.

    This concludes the description of the reduction. Since the post-processing step only takes $1$ round of additional communication, and the bipartite maximal matching problem cannot be solved in $o(\Delta + \log n / \log \log n)$ rounds \cite{Balliu2019}, the same holds for the $2$-bounded stable assignment problem.
\end{proof}

As for any $k \geq 2$, any solution to the $k$-bounded stable assignment problem is also a solution to the $2$-bounded stable assignment problem, the above lower bound also holds for the $k$-bounded stable assignment problem, for all $k \geq 2$.

\begin{theorem}\label{thm: 2boundstable}
    There is a deterministic algorithm that solves the $2$-bounded stable assignment problem in $O(C)$ communication rounds in the \local model.
\end{theorem}
\begin{proof}
    We identify the customers as oriented hyper-edges with one node chosen as the head of the hyper-edge. Consider the following algorithm that makes heavy use of the fact that we only have three different types of load per server. Initialize an arbitrary orientation of the hyper-edges and recall that a happy customer $c$ is such that $c$ cannot locally improve by changing its server. For each server $s$, until all incoming hyper-edges are happy:
    \begin{itemize}
        \item If there is at least one happy incoming hyper-edge, then direct all incoming unhappy hyper-edges to a server with load $0$.
        \item Otherwise, direct all except one incoming unhappy hyper-edge to a server with load $0$.
    \end{itemize}

    Observe that if the in-degree of a server is at least $1$ in some iteration, it never drops below $1$. This is guaranteed since only incoming hyper-edges are redirected and always at least one incoming edge is preserved. Furthermore, since hyper-edges are always directed to nodes of indegree $0$, a server with indegree $1$ never changes its indegree. It follows that a happy hyper-edge never becomes unhappy.

    Now, we want to argue that for an unhappy hyper-edge $c$, in each iteration, the number of adjacent servers with in-degree $0$ decreases by at least $1$. From this, the runtime bound of $O(C)$ communication rounds follows. If $c$ is redirected to an adjacent server $s'$ in some iteration, then due to the above observations, the in-degree of $s'$ never drops below $1$. If $c$ is not redirected, let $s$ be the server that $c$ is directed to. It must be the case that $c$ was the customer chosen in the ``else'' statement and hence, $c$ becomes the only incoming customer of $s$. This implies that $c$ becomes happy. Correctness of the algorithm follows from the fact that it is executed until all customers are happy.
\end{proof}

\section*{Acknowledgements}
We thank Orr Fischer, Juho Hirvonen, and Tuomo Lempi\"ainen for valuable discussions. This project has received funding from the European Union’s Horizon 2020 research and innovation programme under the Marie Sk\l{}odowska-Curie grant agreement No.\ 840605.

\bibliographystyle{plainnat}
\bibliography{articles}

\end{document}